\documentclass[11pt,english]{article}

\usepackage{amsmath,amssymb,amsthm}
\usepackage{multicol}
\usepackage[table]{xcolor}
\usepackage[margin=1in]{geometry}
\usepackage{graphicx,color}
\usepackage{enumitem}
\usepackage{fullpage}
\usepackage[noblocks]{authblk}
\usepackage{tcolorbox}
\usepackage{babel}
\usepackage{wrapfig}
\usepackage{MnSymbol}
\usepackage{mdframed}
\usepackage{mathtools}
\usepackage{floatrow}
\usepackage{multirow}
\usepackage{booktabs}
\usepackage{palatino}

\usepackage{tablefootnote}
\usepackage{thm-restate}
\usepackage{bbding}
\usepackage{pifont}
\usepackage{nicefrac}

\newcommand{\ddim}{d}

\newcommand{\floor}[1]{\ensuremath{\left\lfloor#1\right\rfloor}}

\newcounter{magicrownumbers}

\RequirePackage{algorithm}
\RequirePackage[]{algpseudocode}

\usepackage{algpseudocode}

\graphicspath{{./figures/}}

\usepackage{tablefootnote}

\definecolor{Darkblue}{rgb}{0,0,.8}
\definecolor{Brown}{cmyk}{0,0.61,1.,0.60}
\definecolor{Purple}{cmyk}{0.45,0.86,0,0}
\definecolor{Darkgreen}{rgb}{0.133,0.500,0.133}

\definecolor{MyGreen}{rgb}{0.200,0.500,0.200}
\renewcommand{\emph}[1]{{\color{MyGreen}{\em #1}}}

\usepackage[colorlinks,linkcolor=Darkblue,filecolor=blue,citecolor=blue,urlcolor=Darkblue,pagebackref]{hyperref}
\usepackage[nameinlink]{cleveref}

\usepackage[colorinlistoftodos,prependcaption,textsize=tiny]{todonotes}

\newcommand{\namedref}[2]{\hyperref[#2]{#1~\ref*{#2}}}
\newcommand{\propref}[1]{\hyperref[#1]{property~(\ref*{#1})}}

\newcommand{\eps}{\varepsilon}
\newcommand{\IGNORE}[1]{}

\newtheorem*{theorem*}{Theorem}
\newtheorem{theorem}{Theorem}[section]
\newtheorem{lemma}[theorem]{Lemma}
\newtheorem{definition}[theorem]{Definition}
\newtheorem{claim}[theorem]{Claim}

\newtheorem*{question*}{Question}

\newtheorem*{conjecture*}{Conjecture}
\newtheorem{question}{Question}

\newcommand{\old}[1]{{}}

\hypersetup{
  colorlinks=true,
  linkcolor=red!70!black,
  linktoc=all,
  citecolor=violet,
    pdfpagemode=FullScreen,
}

\title{Light Spanners with Small Hop-Diameter}

\date{}

 \author{Sujoy Bhore\thanks{Department of Computer Science \& Engineering, Indian Institute of Technology Bombay, Mumbai, India.\\ Email: \href{sujoy@cse.iitb.ac.in}{sujoy@cse.iitb.ac.in}}
 \quad
 Lazar Milenkovi\'{c}\thanks{Tel-Aviv University. Email: \href{milenkovic.lazar@gmail.com}{milenkovic.lazar@gmail.com}}}

\begin{document}
\maketitle

\begin{abstract}
\emph{Lightness}, \emph{sparsity}, and \emph{hop-diameter} are the fundamental parameters of geometric spanners. Arya et al. [STOC'95] showed in their seminal work that there exists a construction of Euclidean $(1+\varepsilon)$-spanners with hop-diameter $O(\log n)$ and lightness $O(\log n)$. They also gave a general tradeoff of hop-diameter $k$ and sparsity $O(\alpha_k(n))$, where $\alpha_k$ is a very slowly growing inverse of an Ackermann-style function. The former combination of logarithmic hop-diameter and lightness is optimal due to the lower bound by Dinitz et al. [FOCS'08]. Later, Elkin and Solomon [STOC'13] generalized the light spanner construction to doubling metrics and extended the tradeoff for more values of hop-diameter $k$. In a recent line of work [SoCG'22, SoCG'23], Le et al. proved that the aforementioned tradeoff between the hop-diameter and sparsity is tight for every choice of hop-diameter $k$. A fundamental question remains: \textit{What is the optimal tradeoff between the hop-diameter and lightness for every value of $k$}?

In this paper, we present a general framework for constructing light spanners with small hop-diameter. Our framework is based on \emph{tree covers}. In particular, we show that if a metric admits a tree cover with $\gamma$ trees, stretch $t$, and lightness $L$, then it also admits a $t$-spanner with hop-diameter $k$ and lightness $O(kn^{2/k}\cdot \gamma L)$. Further, we note that the tradeoff for trees is tight due to a construction in uniform line metric, which is perhaps the simplest tree metric. As a direct consequence of this framework, we obtain a tight tradeoff between lightness and hop-diameter for doubling metrics in the entire regime of $k$.
\end{abstract}

\section{Introduction}
Let $M_X = (X, \delta_X)$ be a finite metric space, which can be viewed as a complete graph with vertex set $X$, where the weight of each edge $(u,v) \in \binom{X}{2}$ is equal to the metric distance between its endpoints, $\delta_X(u,v)$. Let $t \ge 1$ be a real parameter and let $H = (X,E)$ be a subgraph of $M_X$ such that $E \subseteq \binom{X}{2}$.
We say that $H$ is a \emph{$t$-spanner} for $M_X$, if for every two points $u$ and $v$ in $X$, it holds that $\delta_H(u,v) \le t \cdot \delta_X(u,v)$, where $\delta_H(u,v)$ denotes the length of the shortest path between $u$ and $v$ in $H$. Such a path is called $t$-spanner path and parameter $t\ge 1$ is called the \emph{stretch factor} (shortly, \emph{stretch}) of $H$.

Spanners for Euclidean metric spaces (henceforth Euclidean spanners) are fundamental geometric structures with numerous applications, such as topology control in wireless networks~\cite{schindelhauer2007geometric}, efficient regression in metric spaces~\cite{gottlieb2017efficient},
approximate distance oracles~\cite{gudmundsson2008approximate}, and more.
Rao and Smith~\cite{rao1998approximating} showed the relevance of Euclidean spanners in the context of other geometric \textsf{NP}-hard problems, e.g., Euclidean traveling salesman problem and Euclidean minimum Steiner tree problem.
Intensive ongoing research is dedicated to exploring diverse properties of Euclidean spanners; see~\cite{althofer1993sparse, arya1995euclidean, BhoreKK0LPT25, BhoreT21a, BhoreT21b, BhoreT22, BCHZST25, chan2009small, das1995new, gudmundsson2002fast, keil1992classes, ES15, rao1998approximating, le2019truly}. In fact, several distinct constructions have been developed for Euclidean spanners over the years, such as
well-separated pair decomposition (WSPD) based spanners~\cite{callahan1993optimal, gudmundsson2002fast}, skip-list spanners~\cite{arya1994randomized},
path-greedy and gap-greedy spanners~\cite{althofer1993sparse, arya1997efficient}, and many more; each such construction found further applications in various geometric optimization problems. For an excellent survey of results and techniques on Euclidean spanners, we refer to the book titled ``Geometric Spanner Networks'' by Narasimhan and Smid~\cite{narasimhan2007geometric}, and the references therein.

Besides having a small stretch, perhaps the most basic property of a spanner is its \emph{sparsity}, defined as the number of edges in the spanner divided by $n-1$, which is the size of an MST of the underlying $n$-point metric.
Chew~\cite{Chew86} was the first to show that there exists an Euclidean spanner with constant sparsity and stretch $\sqrt{10}$. Later, Keil and Gutwin~\cite{keil1992classes} showed that the Delaunay triangulation is in fact a $2.42$-spanner with constant sparsity. Clarkson~\cite{Clarkson87} designed the first Euclidean $(1+\eps)$-spanner for $\mathbb{R}^2$ with sparsity $O(1/\eps)$, for an arbitrary small $\eps>0$; an alternative algorithm was presented by Keil~\cite{keil1988approximating}.
These two papers~\cite{Clarkson87, keil1988approximating} introduced the so-called $\Theta$-graph as a new tool for designing $(1+\eps)$-spanners with sparsity $O(1/\eps)$ in $\mathbb{R}^2$. Ruppert and Seidel~\cite{ruppert1991approximating} later generalized the $\Theta$-graph to any constant dimension $d$, showing that one can construct a $(1+\eps)$-spanner with sparsity $O(\eps^{-d+1})$. Recently, Le and Solomon~\cite{le2019truly} showed that this bound is tight. 

Another fundamental and extensively studied property of a spanner is its \emph{lightness}, 
defined as the ratio of the sum of the edge weights of the spanner to the weight of the MST of the underlying metric.
Das et al.~\cite{das1993optimally} showed that the ``greedy spanner'', introduced by Alth\"ofer et al.~\cite{althofer1993sparse}, has a constant lightness and stretch $1+\eps$ in $\mathbb{R}^3$, for any constant $\eps>0$. This was generalized later by Das et al.~\cite{narasimhan1995new} to $\mathbb{R}^d$, for all $d\in \mathbb{N}$. 
Rao and Smith~\cite{rao1998approximating} showed in their seminal work that the greedy spanner with stretch  $1+\eps$ has lightness $\eps^{-O(d)}$ in $\mathbb{R}^d$ for every constant $\eps$ and $d$. 
After a long line of work, finally in 2019, Le and Solomon~\cite{le2019truly} improved the lightness bound of the greedy spanner to $O(\eps^{-d}\log \eps^{-1})$ in $\mathbb{R}^d$. For metrics with doubling dimension\footnote{The doubling dimension of a metric space $(X, \delta_X)$ is the smallest value $\ddim$ such that every ball $B$ in the metric space can be covered by at most $2^{\ddim}$ balls of half the radius of $B$. A metric space is called doubling if its doubling dimension is constant.} 
$\ddim$, Borradaile et al.~\cite{borradaile2019greedy} showed that the greedy spanner with stretch $1+\eps$ has lightness $\eps^{-2\cdot \ddim}$ (see also~\cite{Gottlieb15} for an earlier work).

Besides having small \emph{stretch} and \emph{sparsity}, a spanner often possesses additional valuable properties of the underlying metric.
One such critical property is the \emph{hop-diameter}: a $t$-spanner for $M_X$ has hop-diameter of $k$ if, for any two points $u, v \in X$, there is a t-spanner path between $u$ and $v$ with at most $k$ edges (or \emph{hops}). Already in 1994, Arya et al.~\cite{arya1994randomized} proposed a construction of Euclidean $(1+\eps)$-spanners with logarithmic hop-diameter and a constant sparsity. In a subsequent work, Arya et al.~\cite{arya1995euclidean} showed that there exists a construction of Euclidean $(1+\varepsilon)$-spanners with hop-diameter $O(\log n)$ and lightness $O(\log n)$. The same paper gives a general tradeoff of hop-diameter $k$ and sparsity $O(\alpha_k(n))$ for $(1+\eps)$-spanners, where $\alpha_k$ is an extremely slowly growing inverse of an Ackermann-style function (see also \cite{BTS94,Sol13}). 
The former tradeoff of hop-diameter versus lightness is optimal due to the lower bound by Dinitz et al.~\cite{dinitz2010low}. 
Later, in 2015, Elkin and Solomon~\cite{ES15} presented a light $(1+\eps)$ spanner construction for doubling metrics and gave a tradeoff between the hop-diameter $k$ and lightness for more values of $k$.
Recent works by Le et al.~\cite{LeMS22, LeMS23} showed that the tradeoff of hop-diameter $k$ and sparsity $O(\alpha_k(n))$ is asymptotically optimal, for every value of $k \ge 1$.
However, despite a plethora of results on tradeoffs between \emph{lightness} and \emph{hop-diameter}, the following question remained open.

\begin{question}\label{q:tradeoff}
Given a set of points in $\mathbb{R}^d$, what is the optimal tradeoff between lightness and hop-diameter $k$, for every value of $k$?
\end{question}

A notion arguably stronger than that of a \emph{spanner} is a \emph{tree cover}. For a metric space $M_X = (X, \delta_X)$ let $T = (V_T, E_T)$ be a tree with $X \subseteq V_T$. We say that the tree $T$ is \emph{dominating}, if for every $u, v \in X$, it holds that $\delta_T(u,v) \ge \delta_X(u,v)$.
A \emph{tree cover} with \emph{stretch} $t$ is a collection of dominating trees such that for every pair of vertices $u,v \in X$, there exists a tree $T$ in the collection with $\delta_T(u,v) \le t \cdot \delta_G(u,v)$. The \emph{size} of a tree cover is the number of trees in it. The lightness of a tree in the cover is the ratio of its weight to the weight of an MST of the underlying metric space.
The \emph{lightness} of a tree cover is the maximum lightness among the trees in the cover. We use $L$-light $(\gamma, t)$-tree cover to denote a tree cover with lightness $L$, $\gamma$ trees, and stretch $t$. 
Clearly, the union of all the trees in a $L$-light $(\gamma,t)$-tree cover constitutes a $t$-spanner with sparsity bounded by $O(\gamma)$ and lightness bounded by $O(\gamma L)$.

The aforementioned tradeoff between hop-diameter and sparsity for Euclidean $(1+\eps)$-spanners by Arya et al.~\cite{arya1995euclidean} was in fact achieved via tree covers.
Their celebrated \emph{``Dumbbell Theorem''} demonstrated that in $\mathbb{R}^d$, any set of points admits a tree cover of stretch $1+\eps$ that uses only $O(\eps^{-d}\cdot \log(1/\eps))$ trees. Later, Bartal et el.~\cite{bartal2022covering}
generalized this theorem for doubling metrics. 
The bound on the tree cover size of \cite{arya1995euclidean} was recently improved by Chang et al.~\cite{ChangC0MST24} by a factor of $1/\eps$. This is tight up to a logarithmic factor in $(1/\eps)$ due to the lower bound on the sparse Euclidean spanners by Le and Solomon~\cite{le2019truly}. There are other tree cover constructions for doubling metrics~\cite{chan2016hierarchical,bartal2022covering}, and for other metrics, such as planar and minor-free~\cite{bartal2022covering,chang2023covering,CCLMST24,gupta2005traveling,KLMN04} and general~\cite{bartal2022covering,mendel2007ramsey,naor2012scale}.\footnote{Metric induced by a graph $G=(V,E)$ is a metric space with point set $V$, where for every $u,v \in V$, their distance in the metric is equal to the shortest path distance between $u$ and $v$ in $G$, denoted by $\delta_G(u,v)$.}

\subsection{Our contributions}
We present a general framework for constructing light spanners with small hop-diameter. Our starting point is the construction of a light $1$-spanner with a bounded hop-diameter for tree metrics. 
The bounds are summarized in the following theorem, with the proof presented in \Cref{sec:ub}.

\begin{restatable}{theorem}{ub}\label{thm:ub}
For every $n \ge 1$, every $k\ge 1$, and every metric $M_T$ induced by an $n$-vertex tree $T$, there is a $1$-spanner for $M_T$ with hop-diameter $k$ and lightness $O(kn^{2/k})$.
\end{restatable}

In order to go from tree metrics to arbitrary metrics, we rely on tree cover theorems. This reduction is summarized in the following corollary, with the proof provided in \Cref{sec:reduction}.

\begin{restatable}{corollary}{reduction}\label{cor:reduction}
Let $n \ge 1$ be an arbitrary integer and let $M_X$ be an arbitrary metric space with $n$ points. If $M_X$ admits an $L$-light $(\gamma,t)$-tree cover, then for any $k\ge 1$, the metric space $M_X$ has a $t$-spanner with hop-diameter $k$ and lightness $O(\gamma\cdot L\cdot k\cdot n^{2/k})$.
\end{restatable}

We note that the reduction from \Cref{cor:reduction} holds for any metric space which admits a light tree cover. To exemplify the reduction, we focus on doubling metrics and provide a general tradeoff between hop-diameter and lightness for $(1+\eps)$-spanners in the following theorem. (See \Cref{sec:reduction} for the proof.)
\begin{restatable}{theorem}{doubling}\label{thm:doubling}
For every $k\ge 1$ and $\eps \in (0,1)$, every $n$-point metric space with doubling dimension $\ddim$ has a $(1+\eps)$-spanner with hop-diameter $k$ and lightness $O(\eps^{-O(\ddim)}\cdot kn^{2/k})$.
\end{restatable}
Next, we compare our result with the aforementioned upper bound by Elkin and Solomon~\cite{ES15}. Since both constructions have a term $\eps^{-O(\ddim)}$ in lightness, we ignore those dependencies for clarity.
The aforementioned construction of Elkin and Solomon \cite{ES15} achieves hop-diameter $O(\log_{\rho}{n} + \alpha(\rho))$ and lightness $O(\rho\cdot \log_{\rho}{n})$, for a parameter $\rho \ge 2$. In other words, the construction has hop-diameter $k' + O(\alpha(\rho))$ and lightness $k'n^{c/k'}$ for $k' \ge 1$, $\rho = n^{c/k'}$ and a constant $c$. Note that this tradeoff does not include values in the regime where hop-diameter is $o(\alpha(n))$. Namely, when $k'=O(\alpha(n))$, then $\rho=n^{\Omega(1/\alpha(n))}$ and the dominant term in hop-diameter is $\alpha(\rho)= \Theta(\alpha(n))$. In addition, the exponent $O(1/k)$ of $n$ in the lightness is not asymptotically tight. Recall that the tradeoff we achieve is hop-diameter $k$ versus lightness of $O(kn^{2/k})$, which holds for every value of $1 \le k \le n$. In other words, we give a fine-grained tradeoff for every value of $k$ while nailing down the correct exponent of $n$, due to the lower bound we discuss next.

We complement our constructions with a lower bound for the  $n$-point \emph{uniform line metric}, which is a set of points on $[0,1]$ with coordinates $i/n$ for $0\le i \le n-1$. Refer to \Cref{sec:lb} for the proof.
\begin{restatable}{theorem}{lb}\label{thm:lb}
For every $n \ge 1$, every $k$ such that $1 \le k \le n$, any spanner with hop-diameter $k$  for the $n$-point uniform line metric must have lightness  $\Omega(kn^{2/k})$.
\end{restatable}

Note that, the lower bound holds for any value of stretch $t \ge 1$. Moreover, Dinitz et al.~\cite{dinitz2010low} previously obtained similar bounds, using a more intricate analysis based on linear programming method and connections to the so-called minimum linear arrangement problem. On the other hand, our proof is based on a rather simple combinatorial argument described in less than two pages. We believe such a simple combinatorial argument will have further implications in designing lower bounds for other related problems. 
The purpose of the lower bound is two-fold. First, since the uniform line metric is conceivably the simplest tree metric, we conclude that our upper bound for tree metrics (\Cref{thm:ub}) is tight. In other words, our reduction (\Cref{cor:reduction}) is lossless. Second, the uniform line metric is also a doubling metric, which means that the tradeoff for the doubling metric (\Cref{thm:doubling}) of hop-diameter $k$ and lightness $O(kn^{2/k})$ is tight.

\section{Upper bound for tree metrics}\label{sec:ub}
This section is dedicated to proving Theorem~\ref{thm:ub}. 
\ub*

We first describe the spanner construction in \Cref{alg:ub}. The construction uses the following well-known result.
\begin{lemma}[\cite{FL22}]\label{lem:split}
Given any integer parameter $\ell > 0$ and an $n$-vertex tree $T$, there is a subset $X$ of at most $\frac{2n}{\ell + 1}-1$ vertices such that every connected component of $T \setminus X$ has at most $\ell$  vertices and at most two outgoing edges towards $X$.
\end{lemma}

Let $n \gets |V(T)|$. If $n \le k$, return the edge set of $T$. If $k = 1$, return the clique on $V(T)$. 
When $k=2$, let $c$ be the centroid of $T$, i.e., a vertex such that every component in $T \setminus \{c\}$ has size at most $n/2$. Let $F$ be the returned set of edges, initialized to an empty set. For every vertex $v$ in $V(T) \setminus \{c\}$, add to $F$ the edge $(c,v)$ of weight $\delta_T(c,u)$. Recurse with $k=2$ on each of the subtrees of $T \setminus \{c \}$. 
When $k\ge 3$, let $\ell = k$ if $n \le 2k^2$, and $\ell = \floor{2n^{2/k}}$ otherwise.
Let $X$ be the subset of $V$ guaranteed by \Cref{lem:split} with parameter $\ell$. Let $T_1, \ldots, T_g$ be the components of $T \setminus X$, each of which neighbors up to two vertices of $X$ and each of which has size at most $\ell$. 
For every connected component $T_i$ of $T \setminus X$, let $u_i$ and $v_i$ be the vertices in $X$ neighboring $T_i$. (It is possible that there is only one such vertex, but that case is handled analogously.) For each vertex $w$ in $T_i$, add to $F$ the edge $(u_i, w)$ of weight $\delta_T(u_i, w)$ and $(v_i, w)$ of weight $\delta_T(v_i, w)$. Recurse on $T_i$ with parameter $k$.
Consider a new tree $T_X$ with a vertex set $X$ and the edge set $E_X$ consisting of: (1) all the edges in $E$ with both endpoints in $X$ and (line 25 in \Cref{alg:ub}); (2) edges $(u_i, v_i)$ for every component $T_i$ which has two neighbors in $X$ (line 34 in \Cref{alg:ub}). Continue recursively for $T_X$ and hop-diameter $k-2$. This concludes the description of the algorithm.

\begin{algorithm}[!ht]
\begin{algorithmic}[1]
\Procedure{Spanner}{$k, T = (V,E)$}
\State $n \gets |V|$
\If {$n \le k$}
\State\Return $E$
\EndIf
\If {$k = 1$}
\State\Return $\{(u, v) \mid u \in V, v \in V\}$  \Comment{Clique on $V$}
\EndIf
\If {$k = 2$}
\State Let $c$ be the centroid of $T$
\State $F \gets \{(c, v) \mid v \in V \setminus \{c\} \}$ \Comment{Connect $c$ to every vertex in $V \setminus \{c\}$}
\State Let $T_1,  \ldots, T_g$ be the components of $T \setminus \{c\}$
\For {$1 \le i \le g$}
\State $n_i \gets |V(T_i)|$
\State $F\gets F \cup \Call{Spanner}{2, T_i}$
\EndFor 
\State\Return $F$
\EndIf
\State $\ell \gets \floor{2n^{2/k}}$
\If {$k=3$} $\ell \gets \floor{n^{2/3}}$
\EndIf
\If  {$k \ge 4$ \textbf{and} $n \le 2k^2$} $\ell \gets k$
\EndIf
\State Let $X$ be the set as in \Cref{lem:split} with parameter $\ell$ and let $T_1, \ldots, T_g$ be the components of $T \setminus X$
\State $E_X \gets \{ (u,v) \mid (u,v) \in E, u \in X, v \in X\}$ \label{lin:ex-1}
\For {$1 \le i \le g$}
\State $n_i \gets |V(T_i)|$
\State $F \gets F \cup \Call{Spanner}{k, T_i}$
\If {$T_i$ has one neighbor $u_i$ in $X$}  $F \gets F \cup \{(u_i, w) \mid w \in V(T_i)\}$
\Else
\State Let $u_i$ and $v_i$ be the two neighbors of $T_i$ in $X$
\State $F \gets F \cup \{(u_i, w) \mid w \in V(T_i)\}$
\State $F \gets F \cup \{(v_i, w) \mid w \in V(T_i)\}$
\State $E_X \gets (u_i, v_i)$ \label{lin:ex-2}
\EndIf
\EndFor
\State $F \gets F \cup  \Call{Spanner}{k-2, T_X = (X, E_X)}$
\State\Return $F$
\EndProcedure
\end{algorithmic}
\caption{Procedure for constructing a spanner of a tree metric induced by a given tree $T= (V,E)$. Parameter $k\ge 1$ is the required hop-diameter. The procedure returns the edge set $F$ of a spanner. The weight of every edge in $F$ is assigned to be equal to the distance of its endpoints in $T$.}\label{alg:ub}
\end{algorithm}

\Cref{lem:stretch-n-hop} asserts that the constructed spanner has stretch $1$ and hop-diameter $k$. \Cref{lem:light} asserts that the constructed spanner has lightness $O(kn^{2/k})$.

\begin{lemma}\label{lem:stretch-n-hop}
For every $k \ge 1, n \ge 1$ and every metric $M_T$ induced by an $n$-vertex tree $T$, procedure \Call{Spanner}{$k, T$} returns a 1-spanner of $M_T$ with hop-diameter $k$. 
\end{lemma}
\begin{proof}
When $n \le k$, the tree $T$ already has hop-diameter $k$. If $k=1$, the procedure construct a clique on $T$ and the lemma holds immediately. (Recall that every edge in $F$ has weight equal to the distance of its endpoints in $T$.)

Next we analyze the case $k=2$. Consider two vertices $u$ and $v$ in $T$ and consider the last recursion level where both $u$ and $v$ were in the same tree, $T'$. The centroid vertex $c'$ of $T'$ is connected via an edge to both $u$ and $v$. Vertex $c'$ is on the shortest path in $T'$ between $u$ and $v$, because after the removal of $c'$, vertices $u$ and $v$ are not in the same subtree anymore by the choice of $T'$. Hence, there is a 2-hop path between $u$ and $v$, consisting of edges $(u,c')$ and $(c', v)$. By construction, the weight of this path is $\delta_T(u,c') + \delta_T(c',v) = \delta_T(u,v)$, where the equality holds because $c'$ is on the shortest path between $u$ and $v$.

It remains to analyze the case $k \ge 3$. Consider two vertices $u$ and $v$ in $T$ and consider the last recursion level where both $u$ and $v$ were in the same tree, $T'$. Let $X'$ be the subset of $V(T')$ that is used in the construction to split $T$ into connected components. Let $T_u$ and $T_v$ be the connected components containing $u$ and $v$, respectively. By the choice of $T'$, the components $T_u$ and $T_v$ are different. Let $u'$ and $v'$ be the vertices in $X'$ such that $u'$ is neighboring $T_u$, $v'$ is neighboring $T_v$ and $u'$ and $v'$ lie on the shortest path between $u$ and $v$. Such a vertex $u'$ exists because all the shortest paths stemming from $T_u$ and going outside of $T_u$ contain one of the at most two vertices in $X$ that neighbors $T_u$. The argument is analogous for $v'$. 
From the construction, we know that the constructed spanner contains edges $(u, u')$ and $(v', v)$. Recall that the vertices in $X$ are connected recursively using a construction for hop-diameter $k-2$. This means that there is a path between $u$ and $v$ with at most $k-2+2=k$ hops. (It is possible that $u' = v'$, but this case is handled similarly.) The stretch of the path between $u'$ and $v'$ is 1 by the induction hypothesis. The weights of edges $(u,u')$ and $(v, v')$ correspond to the underlying distances of their endpoints in $T$. Since $u'$ and $v'$ lie on the shortest path between $u$ and $v$ in $T$, the weight of the spanner path between $u$ and $v$ is equal to their distance in $T$.
\end{proof}

\begin{lemma}\label{lem:light}
For every $k \ge 1, n \ge 1$ and every metric $M_T$ induced by an $n$-vertex tree $T$ of weight $L$, procedure \Call{Spanner}{$k, T$} returns a spanner $H_k$ of $M_T$ with weight $W_k(n, L) = O(kn^{2/k}L)$.
\end{lemma}
The lemma is implied by the following claims.

\begin{claim}\label{clm:ub:n-leq-k}
For every $1 \le n \le k$, $W_k(n, L) \le L$.
\end{claim}
\begin{proof}
The claim is true because the algorithm returns the edge set of $T$ in this case.
\end{proof}

\begin{claim}
For every $n \ge 1$, $W_1(n, L) \le \frac{n^2L}{2}$.
\end{claim}
\begin{proof}
The claim is true because the algorithm returns the clique on the $n$ vertices of $V$. Each edge in the clique has weight at most $L$. The total weight is thus at most $\binom{n}{2}L \le \frac{n^2L}{2}$.
\end{proof}

\begin{claim}
For every $n \ge 1$, $W_2(n, L) \le  nL$.
\end{claim}
\begin{proof}
We use $L_i$ to denote the weight of $T_i$ plus the weight of the edge connecting $T_i$ to $c$.
Clearly, $L = \sum_{i=1}^{g}L_i$. 
From the construction we have that $c$ is connected by an edge to each of the vertices of $T \setminus \{c\}$. The weight of these edges can be upper bounded by $\sum_{i=1}^{g}n_iL_i$, since for each $i \in [g]$, the edge between $c$ and a vertex in $T_i$ has weight of at most $L_i$. The construction proceeds recursively on each $T_i$, and the total weight incurred by recursion is at most $\sum_{i=1}^{g}W_2(n_i, L_i)$. We proceed to upper bound $W_2(n,L)$ inductively. 
\begin{align*}
W_2(n, L) &\le \sum_{i=1}^{g}n_iL_i + \sum_{i=1}^{g}W_2(n_i, L_i)  \\
&\le \sum_{i=1}^{g}n_iL_i + \sum_{i=1}^{g}n_iL_i & \text{induction hypothesis}\\
&= 2\sum_{i=1}^{g}n_iL_i \\
&\le 2\sum_{i=1}^{g}\frac{n}{2}L_i\\
&\le nL
\end{align*}
\end{proof}

\begin{claim}
Consider an invocation of \Call{Spanner}{$k,T$} for $k\ge 3$ and let $L_i$ be the weight of $T_i$ plus the weight of the (at most two) edges connecting $T_i$ to $X$.
Then, $W_k(n,L) \le 2\ell L + W_{k-2}\left(\frac{2n}{\ell+1}, L\right) + \sum_{i=1}^g W_k(\ell, L_i)$.
\end{claim}
\begin{proof}
Clearly, $L = \sum_{i=1}^{g}L_i$. 
Recall that from \Cref{lem:split}, we have $|X| \le \frac{2n}{\ell+1}-1 < \frac{2n}{\ell+1}$.
For each $T_i$, $1 \le i \le g$, the spanner construction adds an edge between each $v \in T_i$ and (at most two) vertices from $X$ which neighbor $T_i$. The total weight of these edges is at most $2\sum_{i=1}^{g}n_iL_i \le 2\ell \sum_{i=1}^gL_i = 2\ell L$. The first inequality holds because each subtree has at most $n_i \le \ell$ vertices. In addition, the vertices in $X$ are connected using a construction with hop-diameter $k-2$. The total weight of the edges used is at most $W_{k-2}(|X|, L) \le W_{k-2}\left(\frac{2n}{\ell+1}, L\right)$.
Finally, each of the components $T_i$ is handled inductively and this contributes at most $\sum_{i=1}^gW_k(n_i, L_i)\le \sum_{i=1}^gW_k(\ell, L_i)$ to the weight.
\end{proof}

\begin{claim}\label{clm:ub:k-3}
For every $n \ge 1$, $W_3(n, L) \le 16n^{2/3}L$.
\end{claim}
\begin{proof}
When $k=3$, parameter $\ell$ is set to $\floor{n^{2/3}}$. 
\begin{align*}
W_3(n,L) &\le 2\ell L + W_{1}\left(\frac{2n}{\ell+1}, L\right) + \sum_{i=1}^g W_3(\ell, L_i)\\
& \le 2n^{2/3}L + W_1(2n^{1/3},L) + \sum_{i=1}^g W_3(\ell, L_i)\\
&\le 4n^{2/3}L + \sum_{i=1}^gW_3(\ell, L_i)\\
&\le 4n^{2/3}L + 16\ell^{2/3}L & \text{induction hypothesis}\\
&\le 4n^{2/3}L + 16 n^{4/9}L \\
&\le 16n^{2/3}L
\end{align*}
The last inequality holds for every $n \ge 4$. 
\end{proof}

\begin{claim} 
For every $1\le n \le 16$, $W_k(n, L) \le 9nL$.
\end{claim}
\begin{proof}
When $k=1$, we have $W_1(n, L) \le \frac{n^2}{2}\cdot L \le 8n L$.
When $k=2$, we have $W_2(n, L) \le nL$.
For $k = 3$ and $1 \le n \le 3$, by \Cref{clm:ub:n-leq-k} we have $W_3(n,L) \le L$. It is straightforward to verify the bound for  $k=3$ and $n \in \{4, 5\}$. For $k=3$ and $n\ge 6$, the bound is implied by \Cref{clm:ub:k-3} because $16n^{2/3} \le 9n$. Finally, when $k\ge 4$, parameter $\ell$ is set to $k$ and the upper bound on $W_k(n,L)$ is obtained as follows.
\begin{align*}
W_k(n,L) &\le 2k L + W_{k-2}\left(\frac{2n}{k}, L\right) + \sum_{i=1}^g W_k(k, L_i)\\
& \le 2k L + 9 \cdot \frac{2n}{k} \cdot L + L\\
& \le 2n L + \frac{18n}{3} \cdot L + L \\
& \le 9n L
\end{align*}
\end{proof}
\begin{claim}
For every $1\le k < n \le 8k$, $W_k(n, L) \le 39kL$.
\end{claim}
\begin{proof}
When $k=1$, we have $W_1(n,L) \le \frac{n^2}{2}\cdot L \le \frac{(8k)^2}{2}\cdot L \le 32kL$.
When $k=2$, we have $W_2(n, L) \le nL \le 8kL$.
When $k=3$, we have $W_3(n, L) \le 16n^{2/3}L \le 16 \cdot (8k)^{2/3}L\le 64kL$.
When $k \ge 4$, we have that $8k \le 2k^2$ and so $\ell = k$.
\begin{align*}
W_k(n,L) &\le 2k L + W_{k-2}\left(\frac{2n}{k}, L\right) + \sum_{i=1}^g W_k(k, L_i)\\
&\le 2k L + W_{k-2}(16,L) + L\\
&\le 2kL + 9\cdot 16L + L\\
&\le 39kL
\end{align*}
\end{proof}

\begin{claim}
For every $k\ge 1$ and $1 \le n \le 2k^2$, $W_k(n, L) \le 41kL$.
\end{claim}
\begin{proof}
When $k=1$, we have $W_1(n,L) = \frac{n^2}{2}\cdot L \le \frac{(2k^2)^2}{2}\cdot L \le 2kL$.
When $k=2$, we have $W_2(n, L) \le nL \le 2k^2L \le 4kL$.
When $k=3$, we have $W_3(n, L) \le 16n^{2/3}L \le 16\cdot (2k^2)^{2/3}L \le 38kL$.
When $k \ge 4$, we have $\ell = k$.
\begin{align*}
W_k(n,L) &\le 2k L + W_{k-2}\left(\frac{2n}{k}, L\right) + \sum_{i=1}^g W_k(k, L_i)\\
& \le 2kL + W_{k-2}(4k, L) + L\\
& \le 2kL + 39(k-2)L + L\\
& \le 41kL
\end{align*}
\end{proof}

\begin{claim}
For all $k \ge 4$ and $n \ge 2k^2$, $W_k(n, L) \le ckn^{2/k}$, for an absolute constant $c$.
\end{claim}
\begin{proof}
We have $\ell = \floor{2n^{2/k}}$.
\begin{align*}
W_k(n,L) &\le 2\ell L + W_{k-2}\left(\frac{2n}{\ell+1}, L\right) + \sum_{i=1}^g W_k(\ell, L_i)\\
&\le 4n^{2/k}L + W_{k-2}\left(n^{\frac{k-2}{k}},L\right) + \sum_{i=1}^g W_k\left(2n^{2/k}, L_i\right)
\end{align*}
\noindent\textit{Case 1: $n < (k/2)^{k/2}$.} Rearranging, we have that $2n^{2/k} < k$. This means that $W_k\left(2n^{2/k}, L_i\right) \le L_i$.
\begin{align*}
W_k(n,L) &\le 4n^{2/k}L + W_{k-2}\left(n^{\frac{k-2}{k}},L\right) + \sum_{i=1}^g W_k\left(2n^{2/k}, L_i\right)\\
&\le 4n^{2/k}L + c(k-2)n^{2/k}L + L\\
&\le ckn^{2/k}L
\end{align*}
The last inequality holds for any $c \ge 5/2$.

\noindent\textit{Case 2: $(k/2)^{k/2}\le n < k^k$.} Rearranging, we have that $2n^{2/k} < k^2$.
\begin{align*}
W_k(n,L) &\le 4n^{2/k}L + W_{k-2}\left(n^{\frac{k-2}{k}},L\right) + \sum_{i=1}^g W_k\left(2n^{2/k}, L_i\right)\\
&\le4n^{2/k}L + c(k-2)n^{2/k}L +\sum_{i=1}^g41kL_i\\
&\le ckn^{2/k}L
\end{align*}
The last inequality is true for any $c \ge 23$.
\noindent\textit{Case 3: $k^k \le n$}
\begin{align*}
W_k(n,L) &\le 4n^{2/k}L + W_{k-2}\left(n^{\frac{k-2}{k}},L\right) + \sum_{i=1}^g W_k\left(2n^{2/k}, L_i\right)\\
&\le 4 n^{2/k}L +c\cdot(k-2)\cdot n^{2/k}L + \sum_{i=1}^g ck(2n^{2/k})^{2/k}L_i\\
&\le 4 n^{2/k}L +c\cdot(k-2)\cdot n^{2/k}L + ck\cdot2^{2/k}\cdot n^{4/k^2}L\\
&=n^{2/k}L\left(4 + c(k-2) +ck\cdot2^{2/k}\cdot n^{\frac{4-2k}{k^2}}\right)\\
&\le n^{2/k}L\left(4 + c(k-2) +c\sqrt{2}\right)\\
&\le ckn^{\frac{2}{k}}L
\end{align*}
The penultimate inequality holds because $k\cdot2^{2/k}\cdot n^{\frac{4-2k}{k^2}} \le \sqrt{2}$ for $k\ge 4$ and $n \ge k^k$. The last inequality holds for any $c \ge 7$.
\end{proof}
\section{Lower bound on the uniform line metric}\label{sec:lb}
This section is dedicated to proving \Cref{thm:lb}, restated here for convenience.
\lb*

Due to the inductive nature of the proofs, we consider a generalization of the uniform line metric. 
\begin{definition}
Let $n \ge 1$ and $1 \le p \le n$ be arbitrary integers. A $(p, n)$ line metric is a set of $p$ points on $[0,1]$ such that for every $0 \le i \le n-1$, the interval $[i/n, (i+1)/n)$ contains at most one point.
\end{definition}

The uniform line metric with $n$ points is an $(n,n)$ line metric. The proof of \Cref{thm:lb} follows from \Cref{st:lb-1} for $k=1$, \Cref{st:lb-2} for $k=2$, and \Cref{st:lb-k} for $k \ge 3$.
\begin{lemma}\label{st:lb-1}
For every $n$ and $p$ such that $1 \le p \le n$, let $M$ be an arbitrary $(p,n)$ line metric. Then, any spanner with hop-diameter $1$ for $M$ has lightness at least $W_1(p,n) \ge  \frac{1}{64}\left(\frac{p^2}{n}\right)^2$.
\end{lemma}
\begin{proof}
Partition $M$ into four consecutive parts, $M_1$, $M_2$, $M_3$, and $M_4$, each consisting of $p/4$ points. The distance between a point in $M_1$ and a point in $M_4$ is at least $\frac{p}{4}\cdot\frac{1}{n}$ and there is at least $(p/4)^2$ such pairs. Since every pair requires a direct edge, the total weight these edges incur is at least $\frac{p}{4n}\cdot \frac{p^2}{16}\ge \frac{1}{64}\left(\frac{p^2}{n}\right)^2$. 
\end{proof}

\begin{lemma}\label{st:lb-2}
For every $n$ and $p$ such that $1 \le p \le n$, let $M$ be an arbitrary $(p,n)$ line metric. Then, any spanner with hop-diameter $2$ for $M$ has lightness at least $W_2(p,n) \ge \frac{1}{16}\cdot \frac{p^2}{n}$.
\end{lemma}
\begin{proof}
Partition $M$ into four consecutive parts, $M_1$, $M_2$, $M_3$, and $M_4$, each consisting of $p/4$ points. Consider two complementary cases. First, if every point in $M_1$ is incident on an edge that goes to either $M_3$ or $M_4$, the total weight of these edges is at least $\frac{p}{4}\cdot \frac{p}{4n}$. In the complementary case, there is a point $a$ in $M_1$ that is not incident on any edge that goes to $M_3$ or $M_4$. Consider an arbitrary point $b$ in $M_4$. A 2-hop path between $a$ and $b$ has to have the first edge between $a$ and a point in $M_1 \cup M_2$ and the second edge between $M_1 \cup M_2$ and $b$. The weight of the second edge is at least $p/(4n)$. Since every point in $M_4$ induces a different edge, the total weight is at least $\frac{p}{4}\cdot \frac{p}{4n}$. (See \Cref{fig:lb-2} for an illustration.) In conclusion, both of the cases require total weight of $\frac{p}{n}\cdot\frac{p}{16}$ and the lower bound follows.
\end{proof}
\begin{figure}        
\includegraphics[width=\textwidth]{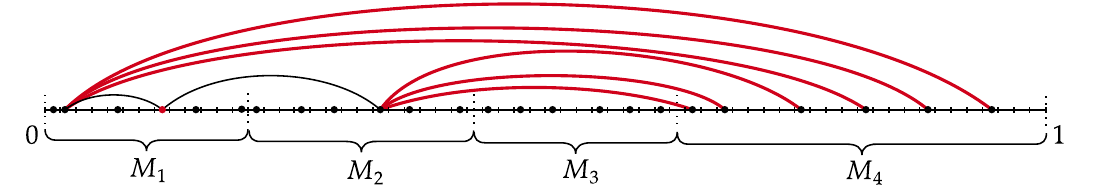}
\caption{An illustration of the proof of the lower bound for $k=2$ (\Cref{st:lb-2}). There is a vertex $p \in M_1$ (highlighted in red) which is not incident on any edge going to $M_3$ or $M_4$. For every point $q \in M_4$, a $2$-hop path from $p$ to $q$ induces a long edge, highlighted in red.}\label{fig:lb-2}
\end{figure}
\begin{lemma}\label{st:lb-k}
For every $n$ and $p$ such that $1 \le p \le n$, let $M$ be an arbitrary $(p,n)$ line metric. Then, for any $k\ge 3$, any spanner with hop-diameter $k$ for $M$ has lightness at least $W_k(p,n) \ge ck\left(\frac{p^2}{n}\right)^{2/k}$, for $c = 1/73728$.
\end{lemma}

\begin{proof}
Let $H_k$ be an arbitrary spanner with hop-diameter $k$ of $M$. 
Partition $M$ into consecutive intervals, each containing $\ell$ points where the value of $\ell$ will be set later. We call an edge \emph{global} if it has endpoints in two different intervals. We call a point \emph{global} if it is incident on a global edge and \emph{non-global} otherwise. 
\begin{figure}
    \centering
    \includegraphics[width=1.08\textwidth]{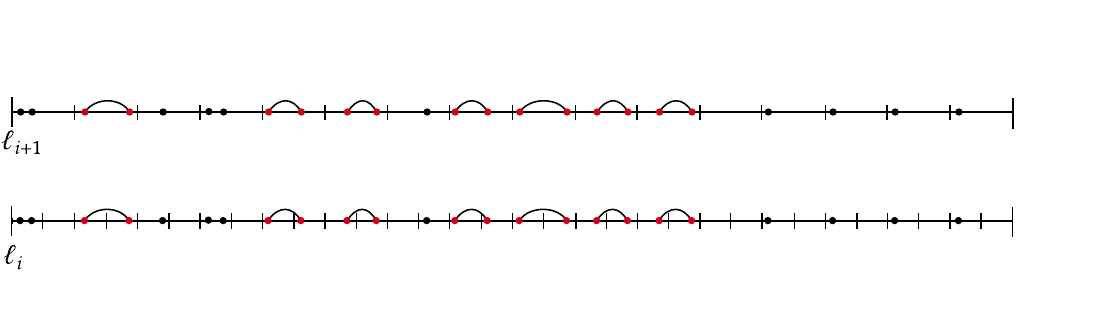}
    \caption{Monotonicity of the number of global points as we increase the interval size from $\ell_i$ to $\ell_{i+1}$. The points highlighted in red were global with respect to $\ell_i$ and became non-global with respect to $\ell_{i+1}$.}
    \label{fig:global}
\end{figure}
\begin{claim}\label{st:global}
If there are $\gamma n$ global points for $\gamma \in [0,1]$, then they contribute at least $\frac{\gamma^2\ell }{16}$ to the total lightness of $H_k$.
\end{claim}
\begin{proof}
Consider an arbitrary interval and recall that it has length $\ell/n$ and consists of $\ell$ points. The interval contains at most $ \frac{\gamma}{4}\cdot \ell$ points that are at distance at most $\frac{\gamma}{4} \cdot \frac{\ell}{n}$ from the left border of the interval. Hence, there are at most $ 2\cdot\frac{\gamma}{4}\cdot \ell$ points that are at distance at most $\frac{\gamma}{4} \cdot \frac{\ell}{n}$ from either of the interval borders. Summing over all the $n/\ell$ intervals, we conclude that there is at most $2\cdot\frac{\gamma}{4}\cdot \ell \cdot \frac{n}{\ell} = \frac{\gamma n}{2}$ points that are at distance at most $\frac{\gamma \ell}{4n}$ from the adjacent interval. The other $\frac{\gamma n}{2}$ points are at distance at least $\frac{\gamma \ell}{4n}$. These points induce edges of total weight of at least $\frac{1}{2}\cdot\frac{\gamma n}{2}\cdot\frac{\gamma \ell}{4n}$, where the factor $1/2$ comes from the fact that each edge might be counted twice.
\end{proof}
Let $\alpha_0 = \frac{1}{4e}$ and let $\ell_0 = \alpha_0n^{2/k}$. We distinguish between two cases. Let $\gamma_0$ be the fraction of global points.\\

\noindent\textit{Case 0.1: $\gamma_0 > \frac{1}{2}$.} Proceed with case analysis below for $i=1$.\\

\noindent\textit{Case 0.2: $\gamma_0 \le \frac{1}{2}$.} In this case, the fraction of the non-global points is $1-\gamma_0 \ge 1/2$. We construct a new line metric, $M'$ by taking a non-global point from every interval containing a non-global point. There are at $n'= \frac{n}{\ell_0}$ intervals and at least $p'=(1-\gamma_0)\frac{n}{\ell_0}$ of them contain a non-global point. (We ignore the rounding issues for simplicity of exposition.) Consider an interval $A$ containing a non-global point $a$ and an interval $B$ containing a non-global point $b$. Any $k$-hop path between $a$ and $b$ has to have the first edge inside of the interval $A$, towards some point $a'$ and the last edge inside of interval $B$ towards some point $b'$. Hence, $a'$ and $b'$ have to be connected via a $(k-2)$-hop path. This is true for any pair of non-global points in $M'$. We use the induction hypothesis for $k-2$ to lower bound the total weight of $H_k$.

\begin{align*}
W_{k-2}(p', n') &\ge c(k-2)\left(\frac{\left(\frac{n}{2\ell_0}\right)^2}{\frac{n}{\ell_0}} \right)^{\frac{2}{k-2}}\\
&\ge c(k-2)\left(\frac{n}{4\alpha_0 n^{2/k}} \right)^{\frac{2}{k-2}}\\
&\ge c(k-2)\left(\frac{1}{4\alpha_0}\right)^{\frac{2}{k-2}}n^{2/k}\\
&\ge c(k-2)e^{\frac{2}{k-2}}n^{2/k}\\
&\ge ckn^{2/k}\\
&\ge ck\left(\frac{p^2}{n}\right)^{2/k}
\end{align*}

The following case analysis consists of cases $i$.1 and $i$.2 for $1 \le i \le \floor{\log_{32}(ek)}=i'$. We let $\alpha_i = \frac{32^i}{4e}$ and $\ell_i = \alpha_i \cdot n^{2/k}$. We use $\gamma_i$ to denote the number of global points with respect to $\ell_i$.
(See Figure~\ref{fig:global} for an illustration of the monotonicity of the number of global points.)\\

\noindent\textit{Case $i$.1: $\gamma_i > \gamma_{i-1}-\frac{1}{4^{i}}$.} Proceed with $i+1$.\\

\noindent\textit{Case $i$.2: $\gamma_i \le \gamma_{i-1}-\frac{1}{4^{i}}$.} We have that $\gamma_{i-1}-\gamma_i \ge \frac{1}{4^{i}}$. This means that $\gamma =\gamma_{i-1}-\gamma_i$ fraction of the points are global with respect to $\ell_{i-1}$ and are not global with respect to $\ell_i$. Their total contribution, by \Cref{st:global}, is at least $\gamma^2\ell_{i-1}/16$. Similarly to Case $0$.2 above, we employ the induction hypothesis for $(k-2)$ to lower bound the contribution of the non-global points. We construct a new line metric, $M'$ by taking a non-global point from every interval containing a non-global point. There are at $n'= \frac{n}{\ell_i}$ intervals and at least $p'=(1-\gamma_i)\frac{n}{\ell_i}$ of them contain a non-global point. Interconnecting the points in $M'$ contributes at least $W_{k-2}(p',n')$ to the total weight of $H_k$. The total weight of $H_k$ can be lower bounded as follows.
\begin{align*}
\allowdisplaybreaks
\frac{\gamma^2\ell_{i-1}}{16}+W_{k-2}(p',n') &\ge 
\left(\frac{1}{4^i}\right)^2\cdot \frac{\ell_{i-1}}{16}+c(k-2)\left(\frac{(p')^2}{n'}\right)^{\frac{2}{k-2}}\\
&\ge \left(\frac{1}{4^i}\right)^2\cdot \frac{\frac{32^{i-1}}{4e}n^{2/k}}{16}+c(k-2)\left((1-\gamma_i)^2\cdot \frac{n}{\ell_i}\right)^{\frac{2}{k-2}}\\
&\ge \frac{8^{i-1}n^{2/k}}{256e}+c(k-2)\left(\frac{1}{4}\cdot \frac{n}{\alpha_i n^{2/k}} \right)^{\frac{2}{k-2}}\\
&\ge \frac{8^{i-1}n^{2/k}}{256e}+c(k-2)\left( \frac{e}{32^i}\right)^{\frac{2}{k-2}}n^{2/k}\\
&\ge n^{2/k}\left(\frac{8^i}{256e} + c(k-2)\left(1+\frac{2}{k-2}\ln{\frac{e}{32^i}}\right) \right)\\
&\ge n^{2/k}\left(\frac{8^i}{256e} +c(k-2) + 2c \ln{\frac{e}{32^i}}\right)\\
&\ge ckn^{2/k} &\text{ for } c = 1/73728\\
\end{align*}
Using that $p \le n$, we have $ckn^{2/k} \ge ck\left(p^2/n\right)^{2/k}$, as required.\\

Finally, we consider the cases $i'$.1 and $i'$.2, where $i'=\floor{\log_{32}(ek)}$.\\

\noindent\textit{Case $i'$.1: $\gamma_{i'} > \gamma_{i'-1}-\frac{1}{4^{i'}}$.}
Observe that $\gamma_{i'}\ge \frac{1}{2} - \frac{1}{3}=\frac{1}{6}$, since $\gamma_0 \ge \frac{1}{2}$ and for every $1 \le i \le i'$ it holds that $\gamma_i > \gamma_{i-1}-\frac{1}{4^{i}}$. By \Cref{st:global}, the contribution of the global points is at least $\gamma_{i'}^2\ell_{i'}/16$, which can be lower bounded as follows.

\begin{gather*}
\frac{\gamma_{i'}^2\ell_{i'}}{16}\ge 
\frac{\ell_{i'}}{16\cdot36}=
\frac{\alpha_{i'}}{16\cdot 36}\cdot n^{2/k}\ge
\frac{32^{i'}}{16\cdot 36\cdot 4e}\cdot n^{2/k}\ge\\
\frac{32^{\log_{32}(ek)-1}}{16\cdot 36\cdot 4e}\cdot n^{2/k}\ge
\frac{1}{73728}\cdot k n^{2/k}\ge
ck\left(\frac{p^2}{n}\right)^{2/k}
\end{gather*}\\

\noindent\textit{Case $i'$.2: $\gamma_i' \le \gamma_{i'-1}-\frac{1}{4^{i'}}$.} Same as case $i$.2 above.

\end{proof}
\section{From light tree covers to light spanners}\label{sec:reduction}
We start this section by proving \Cref{cor:reduction}, restated here for convenience.
\reduction*
\begin{proof}
We construct the spanner $H_X$ for $M_X$ as follows. Let $\mathcal{T}$ be the $L$-light $(\gamma, t)$-tree cover for $M_X$ as in the statement and let $T$ be a tree in $\mathcal{T}$. From \Cref{thm:ub}, we know that the metric $M_T$ induced by $T$ has a 1-spanner with hop-diameter $k$ and lightness $O(kn^{2/k})$ with respect to $M_T$. Since $M_T$ has lightness $L$ with respect to $M_X$, the lightness of $H_T$ with respect to $M_X$ is $O(L\cdot kn^{2/k})$. 
Spanner $H_X$ is obtained as the union of $H_T$ for all $T$ in $\mathcal{T}$. The lightness of $H_X$ is $O(\gamma L\cdot kn^{2/k})$, because there are $\gamma$ trees in $\mathcal{T}$ and each tree has lightness $O(L\cdot kn^{2/k})$ with respect to $M_X$. Consider two arbitrary points $u$ and $v$ in $M_X$. Since $\mathcal{T}$ is a tree cover with stretch $t$, there is a tree $T$ in $\mathcal{T}$ such that $\delta_T(u,v) \le t \cdot \delta_X(u,v)$. By construction, $H_T$ is a $1$-spanner with hop-diameter $k$, so there is a $k$-hop path in $H_T$ (and hence in $H_X$) with length at most $t \cdot \delta_X(u,v)$. The corollary follows.
\end{proof}

Next, we use the following construction of light tree cover from \cite{CCLST25}.
\begin{theorem}[Cf. Theorem 1.2 in \cite{CCLST25}]\label{st:light-doubling}
Given a point set $P$ in a metric of constant doubling dimension $d$ and any parameter
$\eps \in (0,1)$, there exists an $\eps^{-O(d)}$-light $O(\eps^{-O(d)}, 1+\eps)$-tree cover for $P$.
\end{theorem}
\Cref{cor:reduction} and \Cref{st:light-doubling}  immediately imply the following theorem.

\doubling*

\bibliographystyle{alpha}
\bibliography{ref}

\end{document}